\documentclass[journal,onecolumn]{IEEEtran}

\usepackage{cite}
\usepackage{graphicx}
\usepackage{amsmath,amsfonts}
\usepackage{algorithmicx}
\usepackage{algorithm}
\usepackage{url}

\usepackage{mathtools}
\usepackage{amsthm}
\usepackage{hyperref}
\usepackage{enumitem}

\newtheorem{Theorem}{Theorem}[section]

\newtheorem{Lemma}[Theorem]{Lemma}

\newtheorem{Corollary}[Theorem]{Corollary}
{\theoremstyle{definition}
\newtheorem{Definition}[Theorem]{Definition}
\newtheorem{Example}[Theorem]{Example}

}

\newcommand{\R}{\mathbb R}
\newcommand{\Z}{\mathbb Z}
\newcommand{\N}{\mathbb N}
\newcommand{\Ls}{\mathrm{L}}
\newcommand{\As}{\mathcal{A}}
\newcommand{\Hs}{\mathcal{H}}
\newcommand{\defd}{\mathrm{def}}
\newcommand{\lnum}{\mathrm{lnum}}

\usepackage{algpseudocode}

\algrenewcommand\algorithmicwhile{\textbf{While}}
\algrenewcommand\algorithmicfor{\textbf{For}}
\algrenewcommand\algorithmicdo{\textbf{Do}}
\algrenewcommand\algorithmicif{\textbf{If}}
\algrenewcommand\algorithmicthen{\textbf{Then}}
\algrenewcommand\algorithmicelse{\textbf{Else}}
\algrenewcommand\algorithmicend{\textbf{End}}
\algrenewcommand\algorithmicreturn{\textbf{Return}}

\begin{document}

\title{Representing Piecewise-Linear Functions\\ by Functions with Minimal Arity}

\author{Christoph~Koutschan, Anton~Ponomarchuk, and~Josef~Schicho%
\thanks{C.~Koutschan is with the Johann Radon Institute for Computational
  and Applied Mathematics, Linz, Austria}%
\thanks{A.~Ponomarchuk is with the Johann Radon Institute for Computational
  and Applied Mathematics, Linz, Austria}%
\thanks{J.~Schicho is with the Research Institute for Symbolic Computation,
  Johannes Kepler University, Hagenberg im Mühlkreis, Austria.}%
\thanks{The research reported in this paper has been partly funded by BMK, BMDW,
  and the Province of Upper Austria in the frame of the COMET Programme managed
  by FFG in the COMET Module S3AI.}
}

\markboth{}%
{Koutschan, Ponomarchuk, Schicho:
  Representing Piecewise-Linear Functions by Functions with Minimal Arity}


\maketitle

\begin{abstract}
Any continuous piecewise-linear function $F\colon \R^{n}\to \R$ can be represented 
as a linear combination of maxima of at most $n+1$ affine-linear functions. 
This upper bound is sharp, that is, for every $n$ there exists such
a function that is not a linear combination of maxima of $n$ arguments.
In this paper, we give a method to determine, for any given such~$F$, the minimal number $k$ such
that $F$ can be written as a linear combination of maxima of $k$ arguments.
\end{abstract}

\begin{IEEEkeywords}
Continuous piecewise linear function, neural network, machine learning, ReLU activation function.
\end{IEEEkeywords}

\section{Introduction}
A function
$F\colon\R^{n}\to\R$ is continuous piecewise-linear (CPWL) if it can be written as a composition of affine-linear functions and the $\max$ function. 
It is also true that every CPWL function can be written as a linear combination of maxima of affine-linear functions. For some applications, it is important that the number of arguments of the max summands is small. For instance, \cite{breiman1993hinging} suggests the vector space generated by maxima of \emph{two} affine-linear functions as convenient function space for regression, classification, and functional approximation. Even more importantly, the class of CPWL functions is exactly the class of functions that can be computed by neural networks with ReLU activation function, see~\cite{arora:18}. By using a result from~\cite{wang_sun:05} that any CPWL function~$F$ can be written as a linear combination of maxima of at most $n+1$ affine-linear functions, the authors of~\cite{arora:18} show that $F$ can be computed by a ReLU neural network with $\lceil\log_2(n+1)\rceil+1$ hidden layers. Moreover, in~\cite{BBHSY25}, this upper bound was improved to $\lceil\log_3(n-1)\rceil+1$ hidden layers.

In the case of representing $F$ as a linear combination of maxima, the upper bound of arguments for $\max$ functions is tight, in the sense that there
exist functions which do not allow a representation as a linear combination of maxima with fewer arguments.
This result was obtained independently by two different research groups~\cite{hertrich:23, koutschan:25}. Both
use the witness function $\max(0, x_1, \dots, x_n)$, which cannot be represented as a linear combination of maxima 
with less than $n+1$ arguments.
However, the function $\max(0, x_1, x_2) + \max(0, -x_{1}, -x_{2})$ can be represented 
as a linear combination of maxima with two arguments, see Example~\ref{ex:g_1_g_2_func}. The representation is not unique.
This leads to the question: For a given CPWL function~$F$, what is the minimal number of arguments of maxima that is necessary to represent the function~$F$? 
This paper gives an answer to the question above.

A related problem---not treated in this paper---is to represent a CPWL function as a difference of two convex CPWL functions satisfying a certain optimality criterion.
 This question has been answered in \cite{tran2024minimal} for $n=2$ and partially by~\cite{brandenburg2024decomposition} for $n>2$. Both papers incorporate the duality between the CPWL function~$F$ and polyhedral geometry. For analysing ReLU neural networks and their capacity to represent CPWL functions, tropical geometry has been used, for instance in~\cite{montufar2022sharp, hertrich:23, brandenburg2024real,arora:18, chen:22, haase2023lower}. 

This paper is structured as follows. In Section~\ref{sec:basic_tools}, we provide basic information from polyhedral geometry~\cite{brondsted2012introduction}. Also, we recall the notion of piecewise-constant functions necessary to formulate and prove Theorem~\ref{thm:decomp}, following~\cite{hertrich:23}. In Section~\ref{sec:main_thm}, we introduce the concept of critical flag length, and we prove that a CPWL function can be represented as a linear combination of $\max$ functions with at most $k+1$ arguments if and only if $k$ is bigger than or equal to the critical flag length. We also give an algorithm that computes the critical flag length and a representation with the minimal number of arguments.

\begin{Example}
\label{ex:g_1_g_2_func}
We consider the two functions $F_1(x, y)\coloneqq \max(0, x, y)$ and $F_2(x, y)\coloneqq \max(0,-x,-y)$, that are defined on~$\R^{2}$, see Fig.~\ref{fig:fig_ex1}. 
For the sum of the functions $F_{1}$ and $F_{2}$ the following equality holds:
\begin{multline}
\label{eq:exmpl1_sum_eq}
\max(0, x, y)+  \max(0, -x, -y ) = \\
\max(0, x, y, -x, -y, x-y, y-x).
\end{multline}
By~\cite{MB17,hertrich:23,koutschan:25}, the functions $F_1$ and $F_2$ cannot be represented as a linear combination of $\max$ functions with two arguments. In contrast, the function $F_1+F_2$, see Fig.~\ref{fig:fig_ex2},  can be represented as a linear combination of $\max$ functions with two arguments:
\begin{multline}
\label{eq:exmpl1_eq}
(F_{1}+ F_{2})(x, y) = \\
\max(x, y) + \max(-y, x-y) + \max(-x, y-x).
\end{multline}
It suffices to show that the right-hand sides of Eqs.~\eqref{eq:exmpl1_sum_eq} and~\eqref{eq:exmpl1_eq} are equivalent. Indeed, by applying basic algebraic transformations, we get:
\begin{align}
\label{eq:exmpl1_comb_eq}
&\max(x, y) + \max(-y, x-y) + \max(-x, y-x) =  \nonumber \\ 
&\max(-x, y-x) +\max(0, x, 2x-y, x-y) = \nonumber \\
&\max(0, x, y, -x, -y, x-y, y-x).
\end{align}
Theorem~\ref{thm:decomp} shows why the function $F_{1}+F_{2}$ can be represented as a linear combination of $\max$ with two arguments, and why the functions $F_{1}$ and $F_{2}$ cannot. 
\end{Example}

\begin{figure}[!t]
\begin{center}
\setlength{\tabcolsep}{9pt}
\begin{tabular}{@{}cc@{}}
\includegraphics[height=0.22\textwidth]{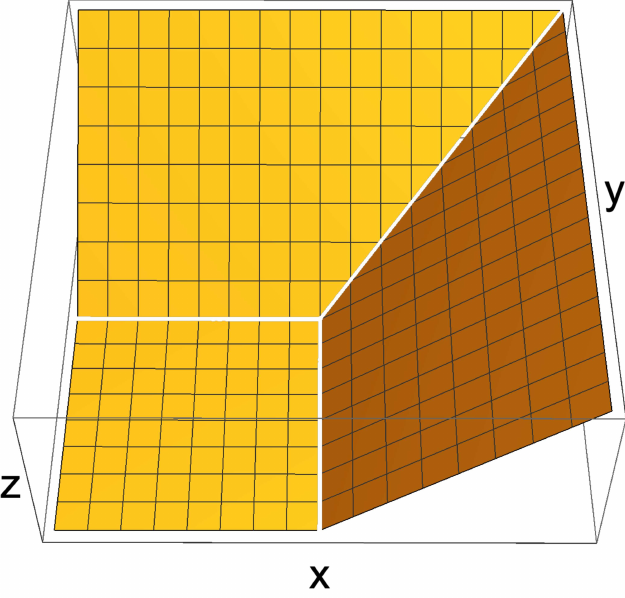} & 
\includegraphics[height=0.23\textwidth]{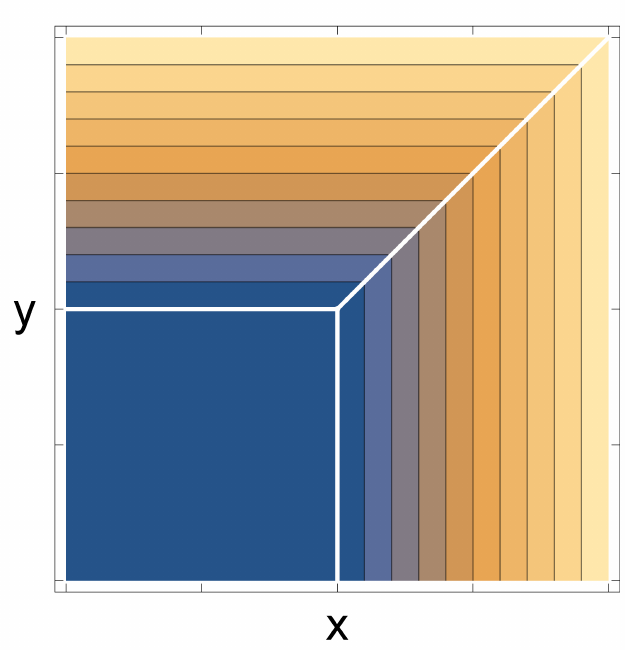} \\[9pt]
\includegraphics[height=0.22\textwidth]{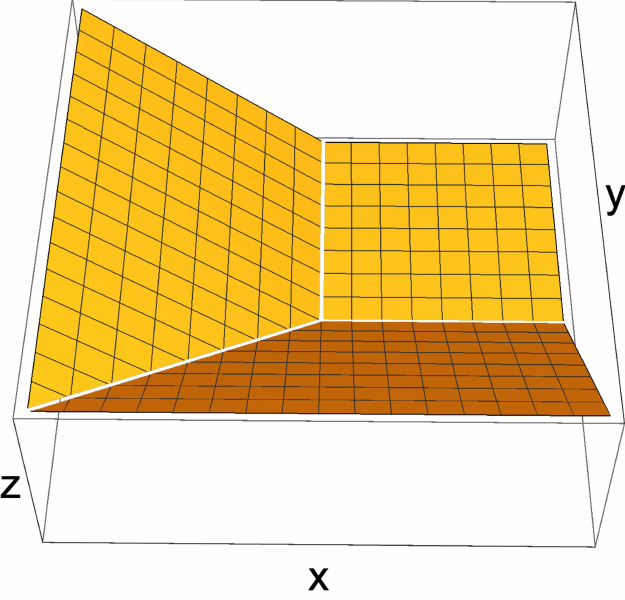} &
\includegraphics[height=0.23\textwidth]{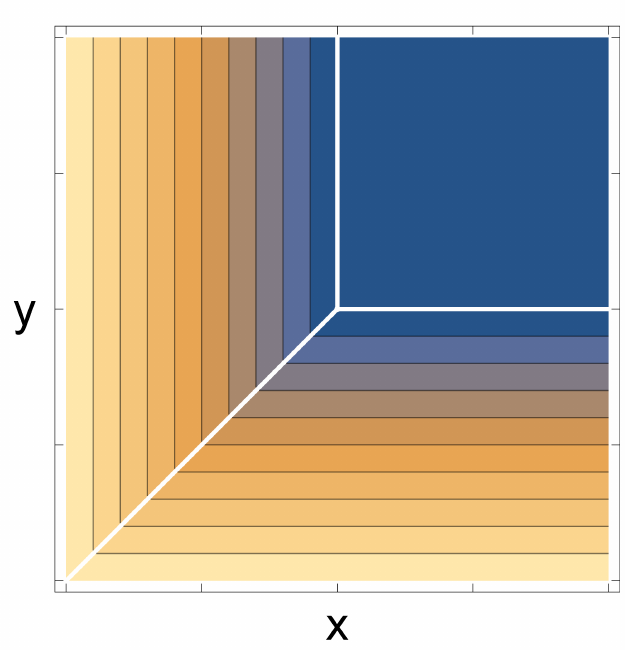}
\end{tabular}
\caption{Two CPWL functions, $F_{1}(x, y) = \max(0, x, y)$ (first row), and $F_{2}(x, y) = \max(0, -x, -y)$ (second row). Each row shows the function $\R^{2}\to \R$ as a three-dimensional plot (first column) and as a contour plot (second column).}
\label{fig:fig_ex1}
\end{center}
\end{figure}

\section{Piecewise-constant functions}
\label{sec:basic_tools}
 Before presenting the main theorem of the paper, we need to recall some basic notions  and introduce the concept of piecewise-constant functions, which
 was applied to analyze convex CPWL functions and their underlying polyhedral complexes~\cite{hertrich:23}. All vector spaces and affine spaces in this paper are assumed to be finite-dimensional.

\begin{Definition}
We say that a property $J$ of points in some non-empty open subset of affine space is defined \emph{almost everywhere} if and only if the set of points for which $J$ does not hold is contained in a finite union of hyperplanes. 
\end{Definition}

\begin{Definition}\label{def:PC}
Let $\Gamma$ be an affine space and $G$ be a vector space. A partial function $f\colon \Gamma \dashrightarrow G$ is \emph{piecewise-constant} (PC) if it is locally constant and if it is defined for almost every point in~$\Gamma$.
\end{Definition}

The set of all points for which a function~$f$ is defined is called the \emph{domain of the definition} of~$f$, denoted by $\defd(f)$.
Note that Definition~\ref{def:PC} implies that the domain of definition of a PC function has finitely many connected components on which the function is constant, and that these components are separated by finitely many hyperplanes.
For an affine space~$\Gamma$, we denote by $\vec{\Gamma}$ the associated vector space of its translations.

 \begin{Definition}
 Let $f\colon \Gamma \dashrightarrow G$ be a PC function defined almost everywhere, then its \emph{lineality space} $\Ls(f)$ is the following set of vectors:
\begin{align*}
  \Ls(f) \coloneqq  \bigl \{ v\in \vec{\Gamma} \;\big|\; & f(x+v) = f(x) \text{ for all $x\in\Gamma$  such that} \\
  & \text{both $f(x)$ and $f(x+v)$ are defined} \bigr\}
\end{align*}
(note that $f(x)$ and $f(x+v)$ are both defined for almost all~$x$).
The dimension of the lineality space (in Lemma~\ref{lemma:LSisVS} we will show that it is indeed a vector space) is denoted by $\lnum(f)\coloneqq \mathrm{dim}(\Ls(f))$ and is called the \emph{lineality number} of~$f$. 
\end{Definition}
 
\begin{Lemma}\label{lemma:LSisVS}
 Let $f\colon \Gamma \dashrightarrow G$ be a PC function, then the set $\Ls(f)$ is a linear subspace of~$\vec{\Gamma}$.
\end{Lemma}
\begin{proof}
 By definition, the vector~$0$ belongs to~$\Ls(f)$. Let $v, w\in \Ls(f)$, then for almost every point $x\in \Gamma$ we have
 \begin{equation*}
     f(x + v + w) = f(x + v) = f(x),
 \end{equation*}
 i.e., $v+w\in \Ls(f)$, and hence $\Ls(f)$ is closed under addition. Similarly, one can argue that also $-v$ is in~$\Ls(f)$. So far, we have established that $kv\in\Ls(f)$ for any $k\in\Z$, and it remains to prove that also $\lambda v\in \Ls(f)$ for any $\lambda\in\R^+$.
 
 Observe that for almost every $x\in\Gamma$, there exists $k\in\Z$ such that $f(x+kv)=f(x+\mu v)$ for all $\mu>k$: we just have to avoid hyperplanes that bound infinite components of $\defd(f)$. By choosing $\mu\coloneqq k+\lambda$ we obtain the following equality:
 \[
   f(x) = f(x+kv) = f(x+\mu v) = f(x+kv+\lambda v) = f(x+\lambda v),
 \]
 from which it follows that $\lambda v\in \Ls(f)$ for any $\lambda\in\R^+$.
\end{proof}
  
For any point~$x$ in an affine space~$\Gamma$ and $\varepsilon > 0$, we denote by $B_{\varepsilon}(x)$ the open ball with center~$x$ and radius~$\varepsilon$, i.e., $B_{\varepsilon}(x)\coloneqq \{y \in \Gamma\mid ||x -y|| < \varepsilon\}$.
An \emph{oriented hyperplane} in $\Gamma$ is a hyperplane $H$ together with the two open half-spaces $H^+$ and~$H^-$.
For instance, if $\Gamma=\R^n$, then
\begin{align*}
  H &= \{x\in \R^{n}\mid a^{T}x = b\}, \\
  H^{+} &= \{x\in \R^{n}\mid a^{T}x > b\}, \\
  H^{-} &= \{x \in \R^{n}\mid a^{T}x < b\}
\end{align*}
for some $a\in\R^n$ and $b\in\R$. 
\begin{Definition}
    Let $\Gamma$ be an affine space.
    An  \emph{oriented flag} in $\Gamma$ is a sequence $\Hs\coloneqq(H_{1}, \dots, H_{k})$ of length $k \in \N$ such that $H_1$ is an oriented hyperplane in~$\Gamma$, and for any $i\in \{2,\dots,k\}$, $H_i$ is an oriented hyperplane in~$H_{i-1}$. 
\end{Definition}

\begin{figure}[!t]
\begin{center}
\setlength{\tabcolsep}{9pt}
\begin{tabular}{@{}cc@{}}
\includegraphics[height=0.2\textwidth]{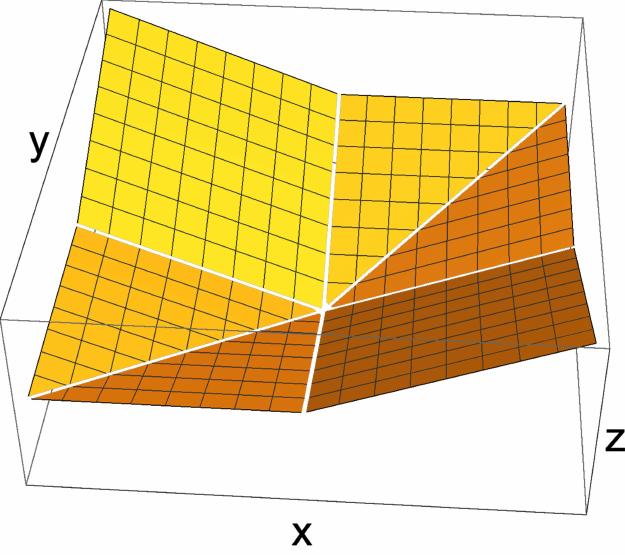} & 
\includegraphics[height=0.21\textwidth]{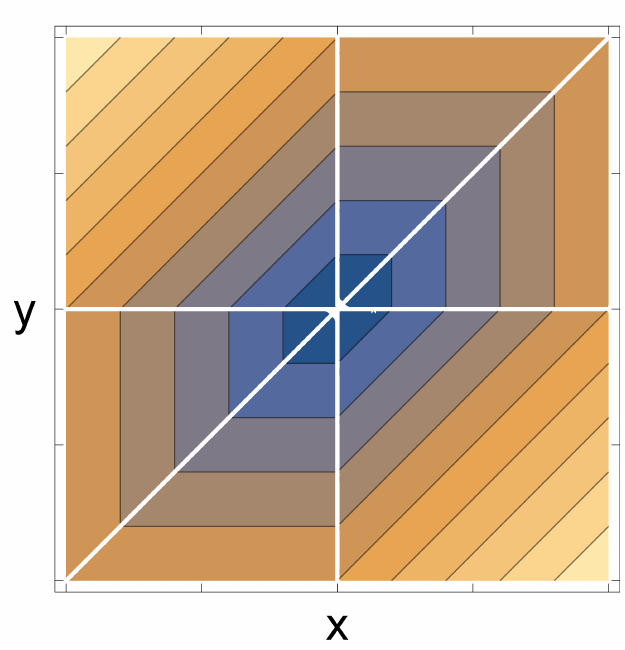}  \\[9pt]
\includegraphics[height=0.2\textwidth]{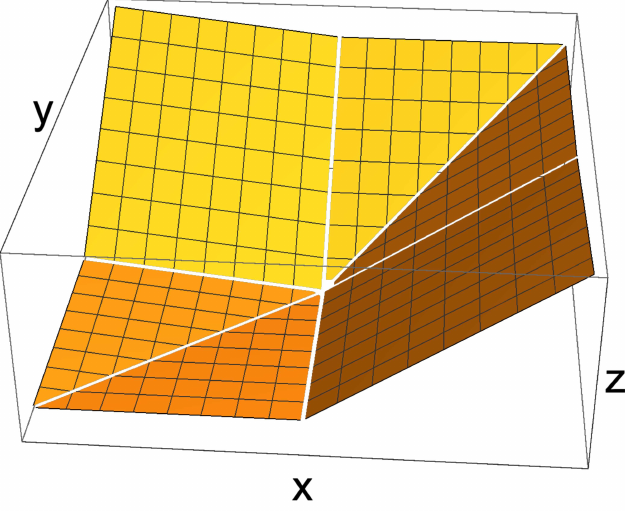} &
\includegraphics[height=0.21\textwidth]{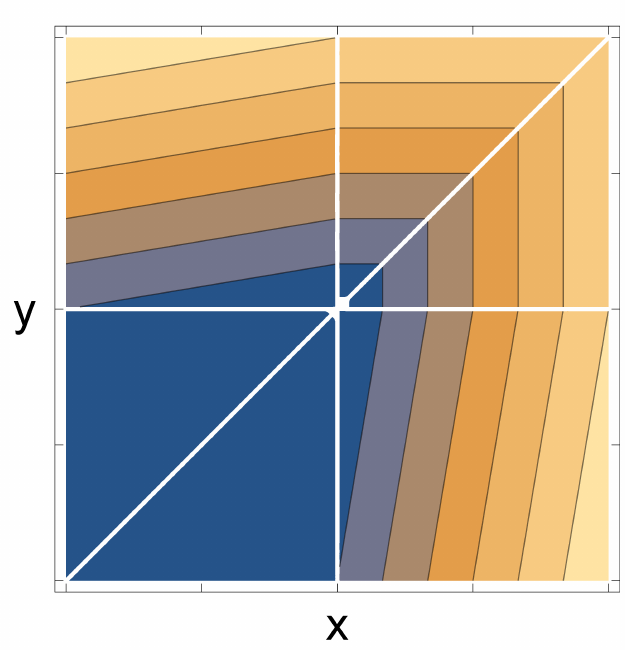}
\end{tabular}
\caption{Two CPWL functions, $F_{3}(x, y) = \max(0, x, y) + \max(0, -x, -y)$ (first row), and $F_{4}(x, y) = 6\max(0, x, y) + \max(0, -x, -y)$ (second row). Each row shows the function $\R^{2}\to \R$ as a three-dimensional plot (first column) and as a contour plot (second column).}
\label{fig:fig_ex2}
\end{center}
\end{figure}

\begin{Definition}
   Let $f\colon\Gamma\dashrightarrow G$ be a PC function and let $H\subset\Gamma$ be an oriented hyperplane. Note that for almost every $x\in H$, there exists $\varepsilon>0$ such that $f$ is defined and constant in the two half balls $B_\varepsilon(x)\cap H^+$ and $B_\varepsilon(x)\cap H^-$. We define the PC function $\Delta_f(H)\colon H\dashrightarrow G$ as the difference of these two constant values, that is, the value in~$H^+$ minus the value in~$H^-$.  
\end{Definition}

\begin{Definition}
 Let $f\colon\Gamma\dashrightarrow G$ be a PC function, and let $\Hs = (H_{1}, \dots, H_{k} )$ be an oriented flag of length $k \in \N$. Let $f_0=f$. For $i=1,\dots,k$, let $f_i\colon H_i\dashrightarrow G$, $f_i\coloneqq \Delta_{f_{i-1}}(H_i)$.
 Then we define $\Delta_f({\cal H})\coloneqq f_k$ and call it \emph{delta function}.
 For the flag $\Hs=()$ of length zero, the delta function is defined to be $f$ itself, i.e., $\Delta_{f}(\Hs) = f$.
 \end{Definition}

\begin{Lemma}
\label{lemma:lnumNumber}
Let $f\colon \Gamma \dashrightarrow G $ be a non-constant PC function and let $H \subset \Gamma$ be an oriented hyperplane. Then for the function~$f$ and the delta function~$\Delta_{f}(H)$, the following inequality holds:
\begin{equation}\label{eq:lnumEq}
  \lnum(\Delta_{f}(H)) \geq \lnum(f) .
\end{equation}
\end{Lemma}
\begin{proof}
We distinguish two cases. 

Case 1: Let $\Ls(f)$ be non-parallel to~$H$.
We pick a vector $v \in \Ls(f)$ with $0<||v||<1$ that is not parallel to~$H$. Then for almost every point $x \in H$ there exists an $\varepsilon > 0$ such that $x + \varepsilon v \in B_{\varepsilon}(x) \cap H^{+}$ and $x- \varepsilon v \in B_{\varepsilon}(x)\cap H^{-}$ and such that $f(x+\varepsilon v)=f(x-\varepsilon v)$. It follows:
\begin{equation*}
    \Delta_f(H)(x) = f(x + \varepsilon v) - f(x -\varepsilon v) = 0.
\end{equation*}
Thus $\Delta_{f}(H)(x)=0$ for almost all $x\in H$ and therefore $\lnum(\Delta_{f}(H)) = \dim(\Gamma)-1$. Since $f$ is not constant, the inequality~\eqref{eq:lnumEq} holds. 

Case 2: Let $\Ls(f)$ be parallel to~$H$. 
Our goal is to show that $\Ls(f)\subseteq \Ls(\Delta_{f}(H))$, which is equivalent to show that for any $v \in L(f)$ and for almost all $x \in H$, the following equality holds:
\begin{equation*}
    \Delta_{f}(H)(x +v) =  \Delta_{f}(H)(x).
\end{equation*}
For almost all $x\in H$ and for any vector $w$ with $0<||w||<1$ that is not parallel to~$H$ and that points towards~$H^{+}$, there exists $\varepsilon > 0$ such that $x+\varepsilon w \in B_{\varepsilon}(x)\cap H^{+}\cap \defd(f)$, $x+ v + \varepsilon w \in B_{\varepsilon}(x + v)\cap H^{+} \cap \defd(f)$, 
$x-\varepsilon w\in  B_{\varepsilon}(x) \cap H^{-} \cap \defd(f)$ and 
$x + v-\varepsilon w\in  B_{\varepsilon}(x + v) \cap H^{-} \cap \defd(f)$. 
It implies that the PC function $f$ is defined for these four points, and using the property that $v \in \Ls(f)$, the following equalities hold:
\begin{align*}
  \Delta_{f}(H)(x) &= f(x +\varepsilon w) - f(x - \varepsilon w) \\
  &= f(x + v+ \varepsilon w) - f(x+ v - \varepsilon w) \\
  &= \Delta_{f}(H)(x + v).
\end{align*}
Hence, $\Delta_{f}(H)(x) = \Delta_{f}(H)(x + v) $, which implies that $v \in \Ls(\Delta_{f}(H))$ and therefore $\Ls(f) \subseteq \Ls(\Delta_{f}(H))$. 
\end{proof}

\begin{Lemma}\label{lem:delta_linearity}
  The $\Delta$ symbol is linear in the index, that is, $\Delta_{f+\lambda g}(H)=\Delta_f(H)+\lambda\Delta_g(H)$ for any two PC functions $f,g$ and real number $\lambda$.
\end{Lemma}
\begin{proof}
  The proof is a direct consequence of the properties of PC functions.
\end{proof}

\begin{Lemma}\label{lemma:deltaLayer}
Let $f\colon\Gamma \dashrightarrow G$ be a PC function such that for all oriented hyperplanes $H$, the delta function $\Delta_{f}(H)$ is equal to~$0$. Then $f$ is constant.
\end{Lemma}
\begin{proof}
Straightforward.
\end{proof}

\begin{Lemma}\label{lemma:deltaPCfunc}
  Let $f\colon \Gamma \dashrightarrow G$ be a PC function. Then the set~$\mathcal{P}$ of all oriented hyperplanes~$H$ such that $\Delta_{f}(H)\not\equiv 0$ is finite.
\end{Lemma}
\begin{proof}
We aim to show that every oriented hyperplane $H\in\mathcal{P}$ contains an open subset of points on which the PC function~$f$ is not defined. By the definition of a PC function, there are only finitely many such hyperplanes, which will imply that $\mathcal{P}$ contains finitely many oriented hyperplanes. 

For any $H \in \mathcal{P}$ there exists $x\in H$ such that $\Delta_{f}(H)(x)\neq 0$. By the definition of the PC function $\Delta_{f}(H)$ for the point~$x$, there exists $\varepsilon > 0$ and a ball $B_{\varepsilon}(x)$ such that for all points $v_1\in B_{\varepsilon}(x)\cap H^{+}$ and $v_{2}\in B_{\varepsilon}(x)\cap H^{-}$ the function~$f$ is defined and $f(v_1)\neq f(v_2)$ holds. Consequently, $f$ is not defined for the set $H\cap B_\varepsilon(x)$.
\end{proof}

By iterating Lemma~\ref{lemma:deltaPCfunc} we get the following corollary:

\begin{Corollary}\label{cor:deltaConstant}
    Let $f\colon  \Gamma \dashrightarrow G$ be a PC function. There are only finitely many flags $\mathcal{H}$ such that the delta function $\Delta_{f}(\mathcal{H})$ is not $0$.
\end{Corollary}

\section{Minimal decomposition}
\label{sec:main_thm}

  In our study of PC functions from $\Gamma$ to~$G$, we are especially interested in the case when both $\Gamma = \R^{n}$ and $G=\R^n$. For a CPWL function $F \colon \R^{n} \rightarrow \R$ we define the PC function~$f$ as follows:
  \begin{equation*}
    f \coloneqq \nabla F.
  \end{equation*}
  We are also interested in the connection between the possible decomposition of the function $F$ as a linear combination of $\max$ functions and the properties of the corresponding PC functions. 
\begin{Definition}
  Let $\Gamma,G$ be vector spaces and let $V$ be a linear subspace of~$\Gamma$. An almost everywhere defined function $f\colon \Gamma\dashrightarrow G$ is called \emph{$V$-invariant} if $f(x+v) = f(x)$ for all $v\in V$ and for all $x\in\Gamma$ such that both $f(x)$ and $f(x+v)$ are defined. (Note that if $f$ is a PC function, then this is equivalent to saying that $V$ is a subspace of $\Ls(f)$.)
\end{Definition}

\begin{Lemma}
\label{lemma:linear_func}
 Let $\Gamma,G$ be vector spaces and let $V$ be a linear subspace of~$\Gamma$. Let $F\colon \Gamma \to G$ be a continuous and almost everywhere differentiable function. If $\nabla F$ is $V$-invariant then $F$ can be written as $F_{1}+ F_{2}$, where $F_{1}$ is $V$-invariant and $F_{2}$ is linear.
\end{Lemma}
\begin{proof}
Let $W$ be a complement of~$V$, i.e., $\Gamma = V \oplus W$, the direct sum of~$V$ and~$W$.
Without loss of generality, we assume that $\Gamma = \R^{\ell}\times \R^{m}$, $V = \R^{\ell}\times 0_{m}$, and $W=0_{\ell} \times \R^{m}$, where $0_{m}$ and $0_{\ell}$ are the zero vectors of length~$m$ and~$\ell$, respectively. 

Before we proceed with the proof of the lemma, we will examine the case in which the function $F$ satisfies the additional property that its restriction to $W$ is equal to~$0$, i.e., $F|_{W} \equiv 0$.
We claim that this stronger assumption implies that the function~$F$ is linear. For any point $y \in\R^{m}$, where $y = (y_1, \dots, y_m)$, we define a \emph{slice function} $F_{y}$ as follows:
\begin{align*}
    F_{y}\colon \R^{\ell} &\to G,\\
    (x_{1}, \dots, x_{\ell}) &\mapsto F(x_1, \dots, x_{\ell}, y_1, \dots, y_{m}). 
\end{align*}
Let $(\R^{\ell})^{\ast}$ be the space of all linear maps from $\R^{\ell}$ to~$G$.
Since $F$ is differentiable almost everywhere, it follows that for almost all~$y$, the slice function $F_{y}$ is differentiable almost everywhere as well. 
Because the derivative $\nabla F$ is assumed to be $V$-invariant, it follows that $\nabla F_{y}$ is $\R^\ell$-invariant. This implies that $\nabla F_{y}$ is constant and hence $F_{y}$ is an affine-linear function, i.e., $F_{y}(x) = M_{y}x + b_{y}$, where $M_{y}\in (\R^{\ell})^{\ast}$ and $b_{y} \in G$. Note that $F_{y}(0)=0$ holds by our additional assumption, which implies that the function $F_{y}$ is actually a linear function, i.e., $F_{y} (x) = M_{y}x$, for almost all~$y$. 

We now argue that this last statement is actually true for all~$y\in\R^m$. For this purpose, let $M\colon\R^m\to (\R^\ell)^\ast$ be the function that maps each $y\in\R^m$ to the linear functional that is defined by $e_i\mapsto F(e_i,y)$ for every unit vector~$e_i$, $i=1,\dots,\ell$, and extends to all $x\in\R^\ell$ by linearity. For those~$y$ for which $F_{y} (x) = M_{y}x$ holds, we have $M_y=M(y)$. Also, the equation $F(x,y)=M(y)x$ is fulfilled almost everywhere, and both sides are continuous, so the equation is fulfilled everywhere. We conclude that $F$ is differentiable at $(x,y)$ if and only if it is differentiable at $(0,y)$. 

We claim that the function $M$ is constant. To show this, note that for almost all points $(x, y)$ the function $F$ is differentiable, which implies that the function $M$ is differentiable for almost all $y\in\R^{m}$. By the assumption that $\nabla F$ is $V$-invariant, we have, for every point $(x, y)$, the differential $\nabla F$ has the same value and the following equality holds: $\nabla F(x, y) = \nabla F(0, y)$. Thereby, for any $i\in \{1, \dots, m\}$ and for almost every $y \in \R^{m}$, we have
\begin{equation*}
  \dfrac{\partial M(y)}{\partial y_i} =
  \dfrac{\partial F(e_i, y)}{\partial y_{i}} =
  \dfrac{\partial F(0, y)}{\partial y_{i}}  =
  \dfrac{\partial M(y)}{\partial y_{i}}\cdot 0 = 0.
\end{equation*}
Hence $M$ is constant and $F$ is linear.

The remaining part of the proof is dedicated to the lemma that has initially been stated. Let $F\colon \Gamma \to G$ be an almost everywhere differentiable function such that $\nabla F$ is $V$-invariant. As above, we denote by $W$ a linear subspace that is complementary to~$V$. 
Let $\mu\colon \Gamma \to W$  be the projection operator with kernel equal to~$V$. We define the function $F_{1}\colon F|_{W}\circ \mu$ as the composition of the projection operator $\mu$ with the input function $F$ restricted to~$W$. Since $\mu$ is linear and maps points from the linear subspace $V$ to $0$, it follows that $F_{1}$ is $V$-invariant.
We define $F_2$ as the difference between $F$ and $F_1$, i.e., $F_2\coloneqq F - F_1$. We have to show that $F_2$ is linear.

The  derivative $\nabla F_{2}$ is $V$-invariant, and also the function~$F_{2}$ restricted to $W$ is equal to~$0$, i.e., $F_{2}|_{W}\equiv 0$. As a result, by the first part of the proof, the function $F_{2}$ is linear.
\end{proof}

\begin{Lemma}
\label{lemma:linealSpace}
Let $F \coloneqq \max(g_{1}, \dots, g_{k+1})$, where $k \leq n$, be a CPWL function such that $g_{i} \colon \R^{n} \rightarrow \R$ is an affine-linear function for each $i \in \{1, \dots, k+1\}$. Then, the lineality space for the function $f \coloneqq \nabla F$ has at least dimension $n-k$.
\end{Lemma}
\begin{proof}
Since the following equality holds:
\begin{align*}
  F &=  \max(g_1,\dots,g_{k+1}) \\
    &= \max(0,g_2-g_1,\dots,g_{k+1}-g_1) + g_1,
\end{align*}
 we also have the correspondent equality for the derivatives:
 \begin{equation*}
     f = \nabla\max(0, \overline{g}_{2},\dots, \overline{g}_{k+1}) + c,
 \end{equation*}
 where $c \coloneqq \nabla g_1$, $\overline{g_{i}} \coloneqq g_{i}-g_1$, for all $i \in \{2, \dots, k+1\}$. 
 Thus, the equality holds for the correspondent lineality spaces:
 \begin{equation}
     \lnum(f) = \lnum(\nabla\max(0,\overline{g}_2,\dots, \overline{g}_{k+1})).
     \label{eq:lnumVal}
 \end{equation}
 
Let $A\in\R^{k\times n}$ be the matrix with rows $\nabla \overline{g}_{2}, \dots, \nabla \overline{g}_{k+1}$.
Its kernel $\ker(A)$ is a subset of $\Ls(\nabla\max(0, \overline{g}_{2}, \dots,\overline{g}_{k+1}))$. For any $x\in\mathbb{R}^n$ and $y\in\mathrm{ker}(A)$, the following equality holds:
\begin{equation*}
\max(0, \overline{g}_{2}, \dots\, \overline{g}_{k+1})(x+y) = \max(0, \overline{g}_{2}, \dots, \overline{g}_{k+1}) (x).
\end{equation*}
Since the matrix $A$ has $k$ rows, it follows that $\mathrm{dim}(\mathrm{ker}(A)) \geq n -k$. 
\end{proof}

\begin{Definition}
 For an arbitrary piecewise constant function $f\colon\Gamma\dashrightarrow G$, we say that an oriented flag $\Hs$ is \emph{critical} if $\Delta_f(\Hs)$ is non-zero. The set of critical flags is denoted by~$\mathcal{F}_f$, and the maximal length of a critical flag is called the \emph{critical flag length of~$f$}, denoted by~$\mathrm{cfl}(f)$.
\end{Definition}

Now we have all the necessary parts to formulate and prove the main theorem of the paper:

\begin{Theorem}
\label{thm:decomp}
Let $F\colon \R^{n}\rightarrow \R$ be a CPWL function, let $f \coloneqq \nabla F$, and let $k \in \N$. The function $F$ can be represented as a linear combination of $\max$ with $k+1$ arguments if and only if $k\ge \mathrm{cfl}(f)$.
\end{Theorem}
\begin{proof}
By~\cite{wang_sun:05}, it was shown that the function $F$ can be represented as a linear combination of $\max$ with at most $n+1$ arguments. Thus, we may assume that $k\leq n$.

Firstly, we show the implication $(\Rightarrow)$.
Let $F$ be a CPWL function that can be represented as a linear combination of $\max$ of at most $k+1$ arguments, i.e., there exists a finite set of indices $I\subset \N$ for which the following equality holds:

\begin{equation}
F = \sum_{i \in I}\sigma_{i}F_{i},
\label{eq:funcEq}
\end{equation}
where $\sigma_{i} \in \{-1, 1\}$, $F_{i} \coloneqq \max(g_{i,1}, \dots, g_{i, p_{i}}) $, such that $p_{i}\leq k+1$  and $g_{i, j}\colon \R^{n}\rightarrow \R$ is an affine-linear function for all  $j \in \{1, \dots, p_{i}\}$ and for all $i \in I$. For every $i \in I$, we denote $f_{i}\coloneqq \nabla F_{i}$. 
Because the operator $\nabla$  is linear and by Eq.~\eqref{eq:funcEq}, $f$ is decomposed as follows:
\begin{equation*}
f = \sum_{i \in I}\sigma_{i}f_{i}.
\end{equation*}
Let $\Hs$ be a critical flag of length~$k$.
By Lemma~\ref{lem:delta_linearity}, the function $\Delta_{f}(\Hs)$ is decomposed as follows:
\begin{equation}
    \Delta_{f}(\Hs) = \sum_{i\in I}\sigma_{i}\Delta_{f_i}(\Hs).
    \label{eq:deltaDecom}
\end{equation}
For $i\in I$, we have the following chain of inequalities:
\begin{equation*}
    n-k \leq \lnum(f_i)\leq \lnum(\Delta_{f_i}(\Hs)).
\end{equation*}
The first inequality follows from Lemma~\ref{lemma:linealSpace}, the second from repeated application of Lemma~\ref{lemma:lnumNumber}; the assumption that $\Hs$ is critical permits this repeated application because all intermediate delta functions (except possibly the last one) will be non-constant. Since the function $\Delta_{f_i}(\Hs)$ is defined on an affine subspace of dimension $n-k$, it follows that the function $\Delta_{f_i}(\Hs)$ is constant. Thus, $\Delta_{f}(\Hs)$ is a linear combination of constant functions and is constant as well. Then, $\Delta_f$ is zero for any flag of length bigger than~$k$, and it follows that $k\ge \mathrm{cfl}(f)$.
 
Secondly, to show the implication $(\Leftarrow)$,
we prove the statement: ``$F$ can be written as a linear combination of maxima of $\mathrm{cfl}(\nabla(F))+1$ affine-linear functions'', by induction on the number~$r$ of critical flags. If this number is zero, then the derivative is constant and the function itself is affine-linear, and the statement is true.

Assume that $r\ge 1$. Let $\Hs$ be a critical flag of maximal length $\mathrm{cfl}(f)$.
We construct a CPWL function $G_1\colon \R^{n}\rightarrow\R$ with $g_1 \coloneqq \nabla G_1$ for which the following properties hold:
\begin{enumerate}[label=(\alph{enumi})]
    \item \label{enum:first_prop_g} The function $G_1$ is a linear combination of maxima of at most $\mathrm{cfl}(f)+1$ affine-linear functions.
    \item \label{enum:second_prop_g} Any flag which is critical for $g_1$ is also critical for~$f$.    
    \item \label{enum:third_prop_g}  The PC function $\Delta_{g_1}(\Hs) $ is equal to  the PC function $ \Delta_{f}(\Hs)$.
\end{enumerate}

Before proving the existence of a function~$G_1$ with the declared properties, let us assume that such a function exists. We define a CPWL function $F_{1}\colon \R^{n}\to\R$ as follows:
\begin{equation*}
    F_{1}\coloneq F - G_1.
\end{equation*}
The PC function $f_{1}\coloneqq \nabla F_{1}$ satisfies the following equality:
\begin{equation*}
    f_{1} = f - g_1.
\end{equation*}

Moreover, the set $\mathcal{F}_{f_1}$ is a proper subset of~$\mathcal{F}_f$, where ``subset'' follows from \ref{enum:second_prop_g} and ``proper'' follows from~\ref{enum:third_prop_g}. Therefore, the cardinality of $\mathcal{F}_{f_1}$ is smaller than the cardinality of~$\mathcal{F}_f$. By induction hypothesis, $F_1$ can be written as a linear combination of maxima with at most $\mathrm{cfl}(f_1)+1$ affine-linear functions. Since $\mathrm{cfl}(f_1)\le\mathrm{cfl}(f)$, the induction is finished.

To finish the proof, one is left to show that the function~$G_1$ exists. 
By $\mathcal{N}$ we define the minimal family of hyperplanes such that the PC function~$f$ is defined on its complement.
Let $\Hs \coloneqq (H_1, \dots, H_{\mathrm{cfl}(f)})$ be an oriented flag in~$\mathcal{F}_{f}$.  Let $x_0 \in H_{\mathrm{cfl}(f) }$ be a point for which all the hyperplanes in~$\mathcal{N}$ that contain $x_0$ also contain $H_{\mathrm{cfl}(f)}$ -- this is true for almost all points in $H_{\mathrm{cfl}(f)}$. 
 There exists $\varepsilon > 0$ such that the ball $U\coloneqq B_{\varepsilon}(x_0)$ has empty intersection with any hyperplane $H\in \mathcal{N}$ for which $x_0 \notin H$.
 For the sake of simplicity and without loss of generality, we assume that $x_0 = 0$ and $F(0)= 0$.
 In two steps, we define the function $G_1\colon \R^{n} \to \R$.

 \textbf{Step 1:} For every point $x\in U$, we define $G_1(x) \coloneqq F(x)$. Because \mbox{$F(0)= 0$}, the function $F$, restricted to the ball~$U$, is \emph{positively homogeneous}, that is for every $\lambda > 0$, and $x\in U$ such that $\lambda x \in U$ holds, $G_{1}(\lambda x) = \lambda G_{1}(x)$.

 \textbf{Step 2:} For any point $x\in \R^{n}$, there exists $\lambda \geq 0$ and $y\in U$, such that $x = \lambda y$. In this case, we define $G_1(x)\coloneqq  \lambda G_1(y)$. 
  
The function~$G_1$ is a CPWL function. It defines a PC function $g_1 \coloneqq \nabla G_1$. 
 The remainder of the proof shows that the functions $G_1$ and $g_1$ satisfy properties 
 \ref{enum:first_prop_g}\,--\,\ref{enum:third_prop_g}.

\textbf{Property~\ref{enum:second_prop_g}:}
Let $\As = (A_1, \dots, A_{p})$, where $p\leq \mathrm{cfl}(f)$, be a flag in~$\mathcal{F}_{g_1}$. 
Then $A_{p}$ passes through~$0$.
Hence, there is a nonempty open subset $V\subset U$ so that 
$\Delta_{g_1}({\cal A})$ is defined in $(V\cap A_{p})$. But inside $U$, the two functions $f$ and $g_1$ coincide. So, $\Delta_{g_1}({\cal A})=\Delta_f({\cal A})$ inside $U$ and therefore $ \mathcal{A}\in  \mathcal{F}_{f}$.

\textbf{Property~\ref{enum:third_prop_g}:}
 Since $\Delta_{g_1}(\Hs)$ and $\Delta_{f}(\Hs)$ coincide in $U\cap H_{\mathrm{cfl}(f)}$ and are constant, it follows that $\Delta_{g_1}(\Hs) = \Delta_{f}(\Hs)$.

\textbf{Property~\ref{enum:first_prop_g}:} 
To prove that the function~$G_1$ can be represented as a linear combination of $\max$ with at most $\mathrm{cfl}(f)+1$ arguments, firstly we are going to show that 
the lineality space of $g_1$ contains the vector space $H_{\mathrm{cfl}(f)}$.
The proof is done by contradiction. Assume that there exists a vector $v \in H_{\mathrm{cfl}(f)}$ such that $v \notin \Ls(g_1)$, which implies, that there exist two points $x, x+v \in \defd(g_1)$ for which the PC function $g_1$ is defined and $g_1(x)\neq g_1(x+v)$. 
By the definition of~$g_{1}$ and the fact that it is defined for $x$ and $x+v$, there exists a scalar $\lambda > 0$ such that $\lambda x, \lambda( x + v)\in \defd(g_1)\cap B_{\varepsilon}(0)$. The PC functions~$f$ and~$g_1$ coincide in $B_{\varepsilon}(x_{0})$ which implies that $f(\lambda x)$ and $f(\lambda (x + v))$ are defined and $f(
\lambda x)\neq f(\lambda( x + v))$. The points $\lambda x$ and $\lambda( x +v) $ belong to different open sets on which $f$ is defined and takes different constant values. On the line segment between the points $\lambda x$ and $\lambda( x + v)$ there exists a point $y\coloneqq (1-\gamma)\lambda x + \gamma \lambda( x +v)$, where $\gamma\in (0, 1)$, such that $y \in B_{\varepsilon}(x_{0})$ and $f$ is not defined in~$y$.
So, the point~$y$ belongs to a hyperplane~$H$ in~$\mathcal{N}$. The hyperplane~$H$ has non-trivial intersection set with the ball $B_{\varepsilon}(0)$ (this set contains~$y$). Hence, the hyperplane~$H$ also contains the affine subspace $H_{\mathrm{cfl}(f)}$. Since $0\in H_{\mathrm{cfl}(f)}$, the affine subspaces~$H$ and $H_{\mathrm{cfl}(f)}$ are linear subspaces. It implies that for every scalar $\gamma'\in \R$ the point $\gamma' v + y$ belongs to~$H$ as well. For $\gamma' = \lambda(1-\gamma)$ resp.\ $\gamma' = - \lambda \gamma$ we get that $\lambda(x + v)$ resp.\ $\lambda x$ belong to~$H$, which contradicts the statement that the points $\lambda x$ and $\lambda (x +v)$ belong to the open sets on which $f$ and $g_1$ are defined. Our initial assumption is incorrect and $H_{\mathrm{cfl}(f)}\subseteq \Ls(g_1)$.

The only thing left is to show that the function~$G_1$ can be represented as a linear combination of $\max$ with at most $\mathrm{cfl}(f)+1$ arguments.  
Let $W\subseteq \R^{n}$ be a linear subspace of dimension~$\mathrm{cfl}(f)$, that is complement to $\Ls(g_1)$. 
Let $\mu_{W}\colon \R^{n} \to W$ be the projection operator with kernel $H_{\mathrm{cfl}(f)}$. By Lemma~\ref{lemma:linear_func} and the fact that $\Ls(g_1)$ contains $H_\mathrm{cfl}(f)$, we can construct an $H_\mathrm{cfl}(f)$-invariant function $F_{1}\colon \R^{n}\to \R$ and a linear function $F_{2}\colon \R^{n}\to \R$ such that $G_{1} = F_{1} + F_{2}$. Moreover, we have $F_{1} = G_{1}|_{W}\circ \mu_{W}$. The function $G_{1}|_{W}$ is defined on a $\mathrm{cfl}(f)$-dimensional linear space $W$ and  by~\cite{wang_sun:05} the function $G_{1}|_{W}$ is a linear combination of $\max$ with at most $\mathrm{cfl}(f)+1$ arguments. It implies that the function~$F_{1}$ is a composition of the projection operator~$\mu_{W}$ and a linear combination of $\max$ with at most $\mathrm{cfl}(f)+1$ arguments. Such composition is also a linear combination of $\max$ with at most $\mathrm{cfl}(f)+1$ arguments, and hence this also holds for the function $G_{1} = F_{1} + F_{2}$.
\end{proof}

\begin{Example}
    \label{ex:g_3_g_4_func}
    Let us consider two functions $F_{3}(x, y)\coloneqq \max(0, x, y)+ \max(0, -x, -y)$ and $F_{4}\coloneqq 6\max(0, x, y)+ \max(0, -x, -y)$. These functions are defined on~$\R^{2}$, see Fig.~\ref{fig:fig_ex2}.  
    The functions~$F_{3}$ and~$F_{4}$ split the input space into the same family of open sets. 
    In Example~\ref{ex:g_1_g_2_func}, we have shown that $F_{3}$ can be represented as a linear combination of $\max$ with two arguments. Since $F_{3}$ and $F_{4}$ tessellate the input space identically, one may think that the function~$F_{4}$ can be represented as a linear combination of $\max$ with two arguments as well. However, the function~$F_{4}$ cannot be represented as a linear combination of $\max$ with two arguments.

To show this, we define the two PC functions $f_{3}\coloneqq \nabla F_{3}$ and $f_{4}\coloneqq \nabla F_{4}$. There are in total three hyperplanes on which $f_{3}$ is not defined. The delta function $\Delta_{f_3}$ with respect to those hyperplanes is constant. On any other hyperplane the delta function $\Delta_{f_{3}}$ is~$0$. By Theorem~\ref{thm:decomp} the function~$F_3$ can be represented as a linear combination of $\max$ functions with at most two arguments. We have already seen its decomposition in Example~\ref{ex:g_1_g_2_func}. On the other hand, for the PC function~$f_4$, the delta function $\Delta_{f_4}((H))$, where the hyperplane $H \coloneqq \{(x, y) \in \R^{2}\mid x=y \}$, is non-constant. To see this, let us pick two points $(1, 1)$ and $(-1, -1)$. The values of $\Delta_{f_4}((H))$ in these points are $(-6, 6)$ and $(-1, 1)$ respectively. Since $\Delta_{f_{4}}((H))$ is non-constant, Theorem~\ref{thm:decomp} says that the function $F_4$ cannot be represented as a linear combination of $\max$ with two arguments. 
\end{Example}

The proof of Theorem~\ref{thm:decomp} is algorithmic, it allows us to compute a decomposition of a given CPWL function $F$
into a linear combination of maxima of at most $\mathrm{clf}(\nabla F)+1$ affine linear functions, see Algorithm~Decompose.
The algorithm recursively calls itself; it terminates because in any recursive call, either the number of arguments of the function~$F$ is decreasing 
or the number of arguments stays the same and the number of critical flags is decreasing.
The algorithm also calls reduceMax~\cite{koutschan:25}, an algorithm that computes a decomposition of a given CPWL function $F\colon\R^n\to\R$
into a linear combination of maxima of at most $n+1$ affine-linear functions.

\renewcommand{\thealgorithm}{} 
\begin{algorithm}
\caption{CriticalFlags}
\begin{algorithmic}[1]
\Require PC function $f$
\Ensure Set of Critical Flags
\Statex
\State $\mathcal{F}\coloneqq\emptyset$;
\State $\mathcal{N}$ is the set of hyperplanes where $f$ is not defined almost everywhere;
\For{$H$ in $\mathcal{N}$}
  \State $f_1\coloneqq\Delta_f(H)$; 
  \State $\mathcal{F}_1\coloneqq\mathrm{CriticalFlags}(f_1)$;
  \For{$\mathcal{H}\in \mathcal{F}_1$}
    \State add $(H,\mathcal{H})$ to $\mathcal{F}$;
  \EndFor
\EndFor
\State \Return $\mathcal{F}$.
\end{algorithmic}
\end{algorithm}

\begin{algorithm*}
\caption{Decompose}
\begin{algorithmic}[1]
\Require CPWL function $F\colon \R^n\to\R$
\Ensure List of signed maxima of at most $\mathrm{cfl}(\nabla F)+1$ affine linear functions that sum up to $F$
\Statex
\State $f \coloneqq \nabla(F)$;
\State $\mathcal{F}\coloneqq\mathrm{CriticalFlags}(f)$;
\If{$\mathcal{F}\ne\emptyset$}
  \State pick $\mathcal{H}=(H_1,\dots,H_k)$ from $\mathcal{F}$ of maximal length;
  \If{$k< n$}
    \State pick $x_0\in H_k$ such that any hyperplane through~$x_0$ where $f$ is not defined contains~$H_k$;
    \State $G_1\coloneqq$ the extension by positive homogeneity of the restriction of $F$ to a small neighborhood of $x_0$;
    \State $F_1\coloneqq F-G_1$;
    \State Choose an orthogonal complement $W$ of~$H_k$, and let $\pi\colon\R^n\to W$ be the projection with kernel~$H_k$;
    \State $\Sigma \coloneqq \mathrm{Decompose}(G_1|_W)$;
    \For{$m\in\Sigma$}
      \State $m \coloneqq m\circ \pi$;
    \EndFor
    \State \Return $\Sigma \cup \mathrm{Decompose}(F_1)$;
  \Else
    \State Compute $\mathrm{Decompose}(F)$ using Algorithm reduceMax in \cite{koutschan:25};
    \State \Return $\mathrm{Decompose}(F)$;
  \EndIf
\Else ~ (in this case, $F$ is affine linear)
  \State \Return $\max(F)$;
\EndIf
\end{algorithmic}
\end{algorithm*}

\begin{Example}
We execute the algorithm on the function $F_3\colon\R^2\to\R$, $F_3(x,y)=\max(0,x,y,-x,-y,x-y,y-x)$ (see also Example~\ref{ex:g_1_g_2_func}
and Example~\ref{ex:g_3_g_4_func}). The gradient~$f_3$ is defined everywhere except on the three lines $H_1\colon x=y$, $H_2\colon x=0$,
and $H_3\colon y=0$. For $i=1,2,3$, we obtain $\Delta_{f_3}(H_i)$ is constant, and so the set of critical flags is
$\{(H_1),(H_2),(H_3)\}$. The critical flag length is $k=1$, and so we expect a linear combination of maxima of at most two
affine linear functions.

We pick the flag $(H_1)$, and the point $x_0=(1,1)$. The extension of the local restriction is
$G_1(x,y)=\max(x,y)$. This is already a maximum of two affine linear functions---so we skip the restriction to an orthogonal
complement and recursive call. We get the function
\begin{align*}
  F_1(x,y) &= \max(0,x,y,-x,-y,x-y,y-x)-\max(x,y) \\
  &= \max(0,-x,-y,-x-y).
\end{align*}
Its critical flags are $(H_2)$ and $(H_3)$. We pick the flag $(H_2)$ and the point $x_1=(0,1)$. The extension
of the local restriction is $G_2(x,y)=\max(0,-x)$. Again, this is already a maximum of two affine linear functions.
We compute the difference
\begin{equation*}
  F_2(x,y) = F_1(x,y)-G_2(x,y) = \max(0,-y)
\end{equation*}
and get the final decomposition 
\begin{align*}
  F_3(x,y) &= G_1(x,y)+G_2(x,y)+F_2(x,y) \\
  &= \max(x,y)+\max(0,-x)+\max(0,-y).
\end{align*}
\end{Example}

Theorem~\ref{thm:decomp} has the implication for different representations as well. 
By applying Theorem~\ref{thm:decomp} to the results from~\cite{arora:18}, we get the following corollary:
\begin{Corollary}
\label{coroll:implTHm1}
For $n \geq 3$, a CPWL function $F\colon \R^{n}\to \R$ can be represented as a ReLU-based neural network with depth at most $\lceil \log_{3}(k)\rceil+1$, if for any flag $\Hs$ of length~$k$ the delta function $\Delta_{\nabla F}(\Hs)$ is zero.
\end{Corollary}
\begin{proof}
    The statement holds by applying Theorem~\ref{thm:decomp} in Theorem~2 from~\cite{BBHSY25}.
\end{proof}

Note that the reversed statement of Corollary~\ref{coroll:implTHm1} is an open problem. More drastically, we do not even know if there is a CPWL function that cannot be represented by a ReLU neural network with two hidden layers. See \cite{hertrich:23, haase2023lower, BBHSY25} for an extended discussion.


\bibliographystyle{IEEEtran}
\bibliography{IEEEabrv,main}

\end{document}